\documentclass[pra,twocolumn,showpacs,preprintnumbers,amsmath,amssymb,superscriptaddress]{revtex4}

\usepackage{bm}
\usepackage{amsbsy}
\usepackage{amsmath}
\usepackage{amsfonts}
\usepackage{amsthm}

\usepackage{epsfig}
\usepackage{graphicx}

\usepackage{color}

\newcommand{\comment}[1]{}

\begin{document}

\theoremstyle{plain}
\newtheorem{theorem}{Theorem}
\newtheorem{lemma}[theorem]{Lemma}
\newtheorem{corollary}[theorem]{Corollary}
\newtheorem{conjecture}[theorem]{Conjecture}
\newtheorem{proposition}[theorem]{Proposition}

\theoremstyle{definition}
\newtheorem{definition}{Definition}

\theoremstyle{remark}
\newtheorem*{remark}{Remark}
\newtheorem{example}{Example}

\newcommand{\FF}{\mathbb{F}}
\newcommand{\ZZ}{\mathbb{Z}}
\newcommand{\RR}{\mathbb{R}}
\newcommand{\CPLX}{\mathbb{C}}

\newcommand{\AC}{\mathcal{A}}
\newcommand{\BC}{\mathcal{B}}
\newcommand{\CC}{\mathcal{C}}
\newcommand{\DC}{\mathcal{D}}
\newcommand{\EC}{\mathcal{E}}
\newcommand{\FC}{\mathcal{F}}
\newcommand{\GC}{\mathcal{G}}
\newcommand{\HC}{\mathcal{H}}
\newcommand{\IC}{\mathcal{I}}
\newcommand{\JC}{\mathcal{J}}
\newcommand{\LC}{\mathcal{L}}
\newcommand{\KC}{\mathcal{K}}
\newcommand{\MC}{\mathcal{M}}
\newcommand{\NC}{\mathcal{N}}
\newcommand{\OC}{\mathcal{O}}
\newcommand{\PC}{\mathcal{P}}
\newcommand{\QC}{\mathcal{Q}}
\newcommand{\RC}{\mathcal{R}}
\newcommand{\SC}{\mathcal{S}}
\newcommand{\TC}{\mathcal{T}}
\newcommand{\UC}{\mathcal{U}}
\newcommand{\VC}{\mathcal{V}}
\newcommand{\WC}{\mathcal{W}}
\newcommand{\XC}{\mathcal{X}}
\newcommand{\YC}{\mathcal{Y}}
\newcommand{\ZC}{\mathcal{Z}}
\newcommand{\rH}{\mathrm{H}}
\newcommand{\rT}{\mathrm{T}}

\newcommand{\ad}{^\dagger }			 			  
\newcommand{\ket}[1]{|#1\rangle}                  
\newcommand{\bra}[1]{\left\langle #1 \right|}     
\newcommand{\braket}[2]{|#1\rangle\langle#2|}    	  
\newcommand{\dyad}[2]{\ket{#1}\bra{#2}}           
\newcommand{\ip}[2]{\langle #1|#2\rangle}         
\newcommand{\matl}[3]{\langle #1|#2|#3\rangle}    
\newcommand{\vect}[1]{\boldsymbol{#1}}         

\newcommand{\nn}{\nonumber\\}

\newcommand{\ii}{\mathrm{i}}					  
\newcommand{\ee}{\mathrm{e}}  
\newcommand{\expo}[1]{\mathrm{e}^{#1}} 
\newcommand{\expoi}[1]{\expo{\ii #1}} 

\def\avg#1{\langle #1\rangle }
\def\cites#1{\color{red}XXXX #1 citation needed XXX \color{black} }
\def\mark#1{\color{red}XXXX --- #1 --- XXX \color{black} }
\def\dya#1{|#1\rangle \langle#1|}
\def\inpd#1#2{\langle#1|#2\rangle }
\def\mat#1{\left(\begin{matrix}#1\end{matrix}\right)}
\def\om{\omega }
\def\ot{\otimes}
\def\tr{{\rm Tr}}

\long\def\ca#1\cb{} 

\title{Construction of all general symmetric informationally complete measurements}  
\author{Amir Kalev}\email{amirk@unm.edu}
\affiliation{Center for Quantum Information and Control, MSC07--4220, University of New Mexico, Albuquerque, New Mexico 87131-0001, USA} 
\author{Gilad Gour}\email{gour@ucalgary.ca}
\affiliation{Institute for Quantum Information Science and Department of Mathematics and Statistics, University of Calgary, 2500 University Drive NW, Calgary, Alberta, Canada T2N 1N4} 
\affiliation{Department of Mathematics, University of California/San Diego, 
        La Jolla, California 92093-0112}

\date{\today}

\begin{abstract}
We construct the set of \emph{all} general (i.e. not necessarily rank 1) symmetric informationally complete (SIC) positive operator valued measures (POVMs). In particular, we show that any orthonormal basis of a real vector space of dimension $d^2-1$ corresponds to some general SIC POVM and vice versa. Our constructed set of all general SIC POVMs contains \emph{weak} SIC POVMs for which each POVM element can be made arbitrarily close to a multiple times the identity.  On the other hand, it remains open if for all finite dimensions our constructed family  contains a rank 1 SIC POVM.
 \end{abstract}


\maketitle

\section{Introduction}

The development of coherent quantum technologies depends on the ability 
to evaluate how well one can prepare or create a particular quantum state~\cite{Fla05}. 
Such an evaluation can be carried out by making appropriate quantum measurements on a sequence of 
quantum systems that were prepared in exactly the same way. The goal in quantum tomography to find efficient  quantum measurements for which their statistics determine completely the states on which the measurement is carried out. Such quantum measurements are said to be informationally complete~\cite{Pru77,Bus89}. In the framework of quantum mechanics quantum measurements are represented by positive-operator-valued measures (POVMs)~\cite{Per93}. Informationally complete POVMs have been studied extensively in the last decade~\cite{Fla05,Pru77,Bus89,Per93,Cav02,Fuc04,Sco06,Ren04,Caves,Rene04,Sco09} due to their appeal both from a foundational perspective~\cite{Fuc04} and from a practical perspective for the purpose of quantum state tomography~\cite{Sco06} and quantum key distribution~\cite{Ren04}. 

A particularly attractive subset of informationally complete POVMs are symmetric-informationally-complete (SIC) 
POVMs~\cite{Caves,Rene04,Sco09,Zau99} for which the operator inner products of all pairs of POVM elements are the same.
Most of the literature on SIC POVMs focus on rank 1 SIC POVMs (i.e. all the POVM elements are proportional to rank 1 projectors).
Such rank 1 SIC POVMs have been shown analytically to exist in dimensions $d=1,...,16,19,24,28,35,48$, and numerically for all dimensions $d\leq 67$ (see~\cite{Sco09} and references therein). However, despite the enormous effort of the last years, it is still not known if rank 1 SIC POVMs exist in all finite dimensions.

The interest in rank 1 SIC POVMs stems from the fact that for quantum tomography rank 1 SIC POVMs are maximally efficient at determining the state of the system~\cite{Rene04}. 
On the other hand, rank 1 POVMs have a disadvantage as they completely erase the original state of the system being measured. Moreover, if the system being measured is a subsystem of a bigger composite system, then when a rank 1 POVM is applied to one part of the system it will destroy all the correlations (both classical and quantum) with the other parts. Therefore, there is a tradeoff between efficiency of a measurement (for the purposes of tomography) and non-disturbance (for the purpose of post processing). For this reason, a more `weak' version of SIC POVMs is very useful for tomography that is followed by other quantum information processing tasks. Properties of general SIC POVMs were also studied in~\cite{App07} and more recently in~\cite{Kalev13}. 

In this paper we construct the family of \emph{all} general SIC POVMs and therefore prove that general SIC POVMs exist in all finite dimensions. In particular, we show that the set of all general SIC POVMs is as big as the set of all orthonormal bases of a real vector space of dimension $d^2-1$. For every SIC POVM in the family we associate a parameter $a$ that determines how close the SIC POVM is to a rank 1 SIC POVM. On one extreme value of $a$, our construction shows that weak SIC POVMs with the property that all the operators in the POVM are close to (but not equal to) a constant factor times the identity always exist. On the other extreme, the question whether there exist rank 1 SIC POVMs in all finite dimensions depends on whether the family of general SIC POVMs contains a rank 1 POVM. 

This paper is organized as follows. In Sec.~\ref{sec:prop} we discuss general properties of SIC POVMs. In Sec.~\ref{sec:const} we introduce our construction of all SIC POVMs. Then in Sec.~\ref{sec:gm} we study SIC POVMs constructed from the generalized Gell-Mann matrices. In section Sec.~\ref{sec:rank1} we discuss the connection of our general SIC POVMs with rank 1 SIC POVMs. In Sec.~\ref{sec:d2} we give an example to illustrate our construction for rank 1 SIC POVM, and we end in Sec.~\ref{sec:conc} with discussion. 

\section{Properties of general SIC POVMs}\label{sec:prop}
We denote by $\rH_d$ the set of all $d\times d$ Hermitian matrices. $\rH_d$ can be viewed as
a $d^2$-dimensional vector space over the real numbers, equipped with the Hilbert-Schmidt inner product 
$( A,B)=\tr(AB)$ for $A,B\in\rH_d$.

\begin{definition}\label{def:sic}
A set of $d^2$ positive-semidefinite operators $\{P_{\alpha}\}_{\alpha=1,...,d^2}$ in $\rH_d$ is called a general SIC POVM if\\
\textbf{(1)} it is a POVM: $\sum_{\alpha=1}^{d^2}P_\alpha=I$, where $I$ is the $d\times d$ identity matrix, and\\
\textbf{(2)} it is symmetric:
$\tr(P_{\alpha}^2)=\tr(P_{\beta}^{2})\neq\frac{1}{d^3}$ for all $\alpha,\beta\in\{1,2,...,d^2\}$, and
$\tr(P_\alpha P_\beta)=\tr(P_{\alpha'}P_{\beta'})$ for all $\alpha\neq\beta$ and $\alpha'\neq\beta'$.\\
\end{definition}
\begin{remark}
In the definition above we assumed that  $\tr(P_{\alpha}^{2})\neq\frac{1}{d^3}$ 
since otherwise (see below) all $P_\alpha=\frac{1}{d^2}I$. Moreover, we will see that
conditions \textbf{(1)} and \textbf{(2)} are sufficient to ensure that the set $\{P_\alpha\}$ is a basis
for $\rH_d$ (for rank one SIC POVMs it was shown for example in~\cite{Caves}). 
Thus, SIC POVMs, as their name suggest, are informationally complete.
\end{remark}

Denoting by
$a\equiv \tr(P_{\alpha}^2)$ and $b\equiv\tr(P_\alpha P_\beta)$ for $\alpha\neq\beta$, we have 
$$
d=\tr(I^2)=\sum_{\alpha=1}^{d^2}\sum_{\beta=1}^{d^2}\tr(P_\alpha P_\beta)=d^2a+d^2(d^2-1)b\;.
$$
Thus, the parameters $a$ and $b$ are not independent and $b$ is given by
\begin{equation}\label{b}
b=\frac{1-da}{d(d^2-1)}\;.
\end{equation}
Also the traces of all the elements in a SIC POVM are equal:
$$
\tr(P_\alpha)=\tr(P_\alpha I)=\sum_{\beta=1}^{d^2}\tr(P_\alpha P_\beta)=a+(d^2-1)b=\frac{1}{d}\;.
$$
Thus, $a$ is the only parameter defining the `type' of a general SIC POVM. The range of $a$ is given by
\begin{equation}\label{a}
\frac{1}{d^3}<a\leq\frac{1}{d^2}\;,
\end{equation}
where the strict inequality above follows from the Cauchy-Schwarz inequality
$\tr(P_\alpha)=\tr(P_\alpha I)<\sqrt{da}$. This is a strict inequality since otherwise all $P_\alpha=\frac{1}{d^2}I$.
Note also that $a=1/d^2$ if and only if all $P_\alpha$ are rank one. 

The set $\{P_\alpha\}_{\alpha=1,...,d^2}$ form a basis for $\rH_d$. To see that all $P_\alpha$ are linearly independent
we follow the same lines as in~\cite{Caves}. Suppose there is a set of $d^2$ real numbers $\{r_\alpha\}_{\alpha=1,...,d^2}$   satisfying
$$
\sum_{\alpha=1}^{d^2}r_\alpha P_\alpha =0\;.
$$
Then by taking the trace on both sides we get $\sum_{\alpha=1}^{d^2}r_\alpha=0$. Moreover, multiplying the equation above by $P_\beta$ and taking the trace gives
\begin{align*}
0&=\sum_{\alpha=1}^{d^2}r_\alpha \tr(P_\alpha P_\beta)=\sum_{\alpha=1}^{d^2}r_\alpha\left(a\delta_{\alpha,\beta}+b(1-\delta_{\alpha,\beta})\right)\\
&=(a-b)\sum_{\alpha=1}^{d^2}r_\alpha\delta_{\alpha,\beta}+b\sum_{\alpha=1}^{d^2}r_\alpha=(a-b)r_\beta\;,
\end{align*}
where we have used $\sum_{\alpha=1}^{d^2}r_\alpha=0$. Now, from Eq.~(\ref{b}) it follows that $a=b$ only if $a=1/d^3$.
Since it is not in the domain of $a$ (see Eq.~(\ref{a})), we conclude that $a\neq b$ and therefore all $r_\beta=0$. Thus, general SIC POVMs are informationally complete POVMs.

The dual basis $\{Q_\alpha\}_{\alpha=1,2,...,d^2}$ of an informationally complete POVM $\{P_\alpha\}_{\alpha=1,2,...,d^2}$,
is a basis of $\rH_d$ satisfying
\begin{equation}\label{dual}
\tr(P_\alpha Q_\beta)=\delta_{\alpha,\beta}\;\;\forall\;\alpha,\beta\in\{1,2,...,d^2\}.
\end{equation}
It is simple to check that any density matrix $\rho\in\rH_d$ can be expressed as
\begin{equation}\label{rho}
\rho=\sum_{\alpha=1}^{d^2}p_\alpha Q_\alpha\;,
\end{equation}
where $p_\alpha\equiv\tr\left(P_\alpha\rho\right)$ are the probabilities associated with the informationally complete 
measurement $\{P_\alpha\}$. Thus, the existence of a dual basis to an informationally complete POVM shows that a $d\times d$ density matrix $\rho$ can be viewed as a $d^2$-dimensional probability vector $\vec{p}=(p_1,...,p_{d^2})$. 

In~\cite{Spe08,Fer08} it was shown that the dual basis to a basis of positive-semidefinite matrices can not itself consist of only positive-semidefinite matrices. 
Therefore, the requirement that $\rho$ in~(\ref{rho}) is positive-semidefinite imposes a constraint on the  probability vectors $\vec{p}$ that correspond to a density matrix $\rho$. For example, the vector
$\vec{p}=(1,0,...,0)$ corresponds to $\rho=Q_1$. Thus, if $Q_1$ is not positive-semidefinite then the vector
$\vec{p}=(1,0,...,0)$ does not correspond to a density matrix.
The set of all probability vectors $\vec{p}=(p_1,...,p_{d^2})$ that correspond to density matrices  form a simplex (in the convex set of $d^2$-dimensional probability vectors) that depends on the choice of the POVM basis $\{P_\alpha\}$ and, in particular, its dual.

The calculation of the dual basis for an informationally complete POVM involves in general cumbersome expressions.
However, for a general SIC POVM $\{P_\alpha\}_{\alpha=1,2,...,d^2}$, with a parameter $a$ as above, the dual basis is given by 
$$
Q_\alpha=\frac{d}{ad^3-1}\left[(d^2-1)P_\alpha-(1-da)I\right]\;.
$$
Using the symmetric properties of $P_\alpha$, it is straightforward to check that the $Q_\alpha$s above indeed satisfy Eq.~(\ref{dual}). Note that for a rank one SIC POVM $Q_\alpha=d(d+1)P_\alpha-I$.

\section{Construction of general SIC POVMs}\label{sec:const}
Let $\rT_d\subset\rH_d$ be the $(d^2-1)$-dimensional subspace of $\rH_d$ consisting of all $d\times d$ traceless Hermitian matrices. To construct general SIC POVMs we need the following two Lemmas:
\begin{lemma}\label{le:symOp}
Let $\{R_\alpha\}_{\alpha=1,2,...,d^2-1}$ be a basis of $\rT_d$ such that
\begin{equation}\label{ggg}
\tr(R_\alpha R_\beta)=r\delta_{\alpha,\beta}-\frac{r}{d^2-1}(1-\delta_{\alpha,\beta})\;,
\end{equation}
for some $0<r\in\mathbb{R}$. Then, for any real $t\neq 0$, the set $\{P_\alpha\}_{\alpha=1,2,...,d^2}$
defined by
\begin{align}\label{symOp}
& P_{\alpha}\equiv\frac{1}{d^2}I+t R_\alpha,\;\;{\rm for}\;\;\alpha=1,2,...,d^2-1,\nn
& P_{d^2}\equiv\frac{1}{d^2}I-t\sum_{\alpha=1}^{d^2-1} R_\alpha,
\end{align}
is a symmetric basis of $\rH_d$ according to Definition~\ref{def:sic}.
\end{lemma}
The proof  is given in App.~\ref{app:proof_le:symOp}. We now construct a basis $\{R_\alpha\}_{\alpha=1,2,...,d^2-1}$ of $\rT_d$ that satisfies the conditions in Lemma~\ref{le:symOp}.
Let $\{F_\alpha\}_{\alpha=1,2,...,d^2-1}$ be an orthonormal basis of $\rT_d$, $\tr(F_\alpha)=0$ and $\tr(F_\alpha F_\beta)=\delta_{\alpha,\beta}$, and let 
$$
R_\alpha\equiv xF_\alpha+y\sum_{\beta\neq\alpha} F_\beta
$$
where $x$ and $y$ are some non-zero real constants.
We then get
\begin{equation}\label{tec}
\tr(R_\alpha R_\beta)=x^2\delta_{\alpha,\beta}+2xy(1-\delta_{\alpha,\beta})+y^2\left((d^2-3)+\delta_{\alpha,\beta}\right)
\end{equation}
We would like to find $x$ and $y$ such that $\tr(R_\alpha R_\beta)$ above has the form Eq.~(\ref{ggg}).
We therefore looking for $x$ and $y$ such that for any $\alpha\neq\beta\in\{1,2,...,d^2-1\}$
\begin{equation}\label{o1}
\frac{\tr(R_{\alpha}^{2})}{\tr(R_\alpha R_\beta)}=-(d^2-1)\;.
\end{equation}
From Eq.~(\ref{tec}) we get
\begin{equation}\label{o2}
\frac{\tr(R_{\alpha}^{2})}{\tr(R_\alpha R_\beta)}=\frac{x^2+y^2(d^2-2)}{2xy+y^2(d^2-3)}=\frac{\omega^2+d^2-2}{2\omega+d^2-3}\;,
\end{equation}
where $\omega\equiv x/y$. Comparing Eq.~(\ref{o1}) and Eq.~(\ref{o2}) gives two solutions for $\omega$:
$$
\omega_{\pm}=1-d^2\pm d.
$$
Without loss of generality, we take $y=1$ so that $x=\omega$. Thus, for any real $t\neq 0$, we define:
\begin{align*}
& P_{\alpha,\pm}=\frac{1}{d^2}I+t \left(\omega_\pm F_\alpha+\sum_{\beta\neq\alpha} F_\beta\right)\;\;{\rm for}\;\;\alpha=1,2,...,d^2-1\\
& P_{d^2,\pm}=\frac{1}{d^2}I-t\sum_{\alpha=1}^{d^2-1}\left(\omega_\pm F_\alpha+\sum_{\beta\neq\alpha} F_\beta\right).
\end{align*}
Denoting by $F=\sum_{\alpha=1}^{d^2-1} F_\alpha$, we get
\begin{align}\label{sicpm}
& P_{\alpha,\pm}=\frac{1}{d^2}I +t\Bigl( F-d(d\mp 1) F_\alpha\Bigr)\;\;{\rm for}\;\;\alpha=1,2,...,d^2-1\nn
& P_{d^2,\pm}=\frac{1}{d^2}I+t(1\mp d)F.
\end{align} 
\begin{lemma}\label{le:symOp2}
Let $\{P_{\alpha}\}_{\alpha=1,2,...,d^2}$ be a symmetric basis of $\rH_d$ according to Definition~\ref{def:sic}.  Then, for any real $t\neq 0$, the two sets $\{F_{\alpha,\pm}\}_{\alpha=1,2,...,d^2-1}$
\begin{align*}
F_{\alpha,\pm}\equiv &\pm \frac1{t}\frac1{d(d\pm1 )^2}\Bigl(\frac1{d}I-P_{d^2}+ (1\mp d)P_\alpha  \Bigr)\nn{\rm for}\;\;\alpha=&1,2,...,d^2-1
\end{align*}
are orthonormal bases of $\rT_d$.
\end{lemma}
\begin{proof}
Assuming that the $\{P_{\alpha}\}_{\alpha=1,2,...,d^2}$ are symmetric basis of $\rH_d$, it is straight forward to check that the $F_{\alpha,\pm}$ form an orthonormal operator basis of $\rT_d$, $\tr(F_\alpha)=0$ and $\tr(F_\alpha F_\beta)=\delta_{\alpha,\beta}$.
\end{proof}
Most importantly, the two orthonormal bases $\{F_{\alpha,\pm}\}$ are related to each other by an orthogonal transformation. Given a symmetric basis $\{P_{\alpha}\}_{\alpha=1,2,...,d^2}$ in $\rH_d$, we can write it in two ways
\begin{align*}
P_{\alpha}&=\frac{1}{d^2}I +t\Bigl( F_{+}-d(d-1) F_{\alpha,+}\Bigr)\nn&=\frac{1}{d^2}I +t\Bigl( F_{-}-d(d+1) F_{\alpha,-}\Bigr)\;\;{\rm for}\;\;\alpha=1,2,...,d^2-1\nn
P_{d^2}&=\frac{1}{d^2}I+t(1-d)F_{+}=\frac{1}{d^2}I+t(1+ d)F_{-}.
\end{align*} 
Therefore, without loss of generality, we can express all symmetric operators on $\rH_d$ in a given structure, say $\{P_{\alpha,-}\}$ of Eq.~(\ref{sicpm}).

We are now ready to introduce the construction of all general SIC POVMs. Let $\{F_\alpha\}$ (with ${\alpha=1,...,d^2-1}$) be an orthonormal basis of $\rT_d$; that is, $\tr(F_\alpha)=0$ and $\tr(F_\alpha F_\beta)=\delta_{\alpha,\beta}$ for all $\alpha,\beta\in\{1,2,...,d^2-1\}$. For example, for $d=2$ the normalized three Pauli matrices form such a basis for $\rT_2$. For higher dimensions one can take for example the generalized Gell-Mann matrices. Given such a fixed basis for $\rT_d$, we define two real numbers $t_0$ and $t_1$ as follows. 
Let $F\equiv\sum_{\alpha=1}^{d^2-1}F_\alpha$ and for $\alpha=1,...,d^2-1$ let $\lambda_\alpha$ and $\mu_\alpha$ be the maximum and minimum eigenvalues of $R_\alpha\equiv F- d(d+1) F_\alpha$, respectively. Denote also by $\lambda_{d^2}$ and $\mu_{d^2}$ the maximum and minimum eigenvalues of $R_{d^2}\equiv(d+1)F$, respectively. Note that for all $\alpha$, $\lambda_\alpha$  are positive and $\mu_\alpha$ are negative since $\tr(F)=\tr(F_\alpha)=0$. We denote
$$
t_0\equiv-\frac{1}{d^2}\min_{\alpha\in\{1,...,d^2\}} \left\{\frac{1}{\lambda_\alpha}\right\}\;,\;
t_1\equiv-\frac{1}{d^2}\max_{\alpha\in\{1,...,d^2\}} \left\{\frac{1}{\mu_\alpha}\right\}.
$$  

\begin{theorem}\label{thm:main}
For any non zero $t\in[t_0,t_1]$ the set of $d^2$ operators 
\begin{align}
& P_{\alpha}\equiv\frac{1}{d^2}I+t\left(F- d(d+1) F_\alpha\right)\;\;\forall\alpha\in\{1,2,...,d^2-1\}\nonumber\\
& P_{d^2}\equiv\frac{1}{d^2}I+t(d+1)F \;,\label{main}
\end{align} 
form a general SIC POVM in $\rH_d$. Moreover, any general SIC POVM can be obtained in this way.
\end{theorem}

\begin{proof}
By construction we have $\sum_{\alpha=1}^{d^2}P_\alpha=I$. 
Moreover, the real parameters $t_0$ and $t_1$ have been chosen in such a way to ensure that all $\{P_\alpha\}_{\alpha=1,...,d^2}$ are positive-semidefinite.  Thus, $\{P_\alpha\}$ is a POVM.
Since $\tr(F_\alpha)=\tr(F)=0$ and $\tr(F_\alpha F)=1$ we get for $\alpha,\beta\in\{1,2,...,d^2-1\}$
\begin{align*}
\tr(P_{\alpha}P_{\beta}) & =\frac{1}{d^3}+t^2\tr\left[\left(F-d(d+1)F_\alpha\right)\left(F-d(d+1)F_\beta\right)\right]\\
& =\frac{1}{d^3}+t^2\left[d^2-1-2d(d+1)+d^2(d+1)^2\delta_{\alpha,\beta}\right]\\
&=\frac{1}{d^3}+t^2(d+1)^2\left(d^2\delta_{\alpha,\beta}-1\right)\;.
\end{align*}
Thus, for all $\alpha,\beta\in\{1,2,...,d^2-1\}$ with $\alpha\neq\beta$ 
$\tr(P_{\alpha}P_\beta)=\tr(P_{\alpha}P_{d^2})=1/d^3-t^2(d+1)^2$. Similarly, $\tr(P_{\alpha}^{2})=\tr(P_{d^2})$
for all $\alpha=1,2,...,d^2-1$.

It is therefore left to show that any general SIC POVM can be obtained in this way.
Indeed, isolating $F_\alpha$ from Eq.~(\ref{main}) gives
\begin{equation}\label{f}
F_\alpha=\frac{1}{td(d+1)^2}\left(\frac{1}{d}I+P_{d^2}-(d+1)P_\alpha\right)\;.
\end{equation}
Thus, if $\{P_\alpha\}_{\alpha=1,...,d^2}$ is a general SIC POVM then one can easily verify that the set $\{F_\alpha\}_{\alpha=1,...,d^2-1}$ defined in~(\ref{f}) is orthonormal.
This completes the proof.
\end{proof}

The parameter $a$ associated with our general SIC POVMs is given by
\begin{equation}\label{at}
a=a(t)\equiv\frac{1}{d^3}+t^2(d-1)(d+1)^3
\end{equation}
The maximum value of $a$ is obtained when $t=t_m\equiv\max\{|t_{0}|,t_{1}\}$ and is depending on the choice of the orthonormal basis $\{F_\alpha\}$. Thus, for a given basis $\{F_\alpha\}$, the construction above generates general SIC POVMs for any 
$\frac{1}{d^3}<a\leq a(t_m)$. 
The value of $a(t_m)$ is always bounded from above by $1/d^2$, and it is equal to $1/d^2$ only if the resulting SIC POVM consists of rank 1 operators. We therefore denote 
\begin{equation}\label{amax}
a_{\max}\equiv\max_{\{F_\alpha\}\in\rT_d}a(t_m)\;,
\end{equation}
where the maximum is taken over all orthonormal bases $\{F_\alpha\}$ of $\rT_d$. We say here that  $\{F_\alpha\}$ is an \emph{optimal} basis of $\rT_d$ if its $a(t_m)$ value equals $a_{\max}$. The conjecture that rank 1 SIC POVMs exist in all finite dimensions is equivalent to the conjecture that $a_{\max}$ always equal to $1/d^2$.

\section{SIC POVM from generalized Gell-Mann operator basis}\label{sec:gm}
We now calculate the maximal value of the parameter, $a$, for the case where the operator basis, $\{F_{\alpha}\}$, is the generalized Gell-Mann basis.  The generalized Gell-Mann operators are a set of $d^2-1$ operators which form a basis for $\rT_d$. We can label them by two indexes $n,m$ each taking on integer values, $n,m=1,2,\ldots,d$, such that
\begin{align}\label{gm1}
G_{nm}=\begin{cases}
  \frac1{\sqrt2}(\ket{n}\bra{m}+\ket{m}\bra{n})& \text{for } n<m, \\
   \frac\ii{\sqrt2}(\ket{n}\bra{m}-\ket{m}\bra{n})& \text{for } m<n
    \end{cases}
\end{align}
and
\begin{equation}\label{gm2}
G_{nn}=  \frac1{\sqrt{n(n+1)}}(\sum_{k=1}^n\ket{k}\bra{k}-n\ket{n+1}\bra{n+1}), 
\end{equation}
for $n=1,2,\ldots,d-1$. The pre-factors in Eqs.~(\ref{gm1}) and~(\ref{gm2}) were chosen such that  $\tr(G_{nm}^2)=1$ for all $n,m=1,2,\ldots,d-1$.

We first bound the eigenvalues of $F=\sum_\alpha F_\alpha=\sum_{n,m}G_{nm}$ using the Weyl's inequality~\cite{franklin93}. The inequality states that the eigenvalues of $F=N+V$, where $N$ and $V$ are  $d\times d$ hermitian matrices,  are bounded as 
\begin{equation}\label{Wineq}
n_i+v_d\leq f_i \leq n_i + v_1.
\end{equation}
where $f_1\geq\cdots\geq f_d$,   $n_1\geq\cdots\geq n_d$, and $v_1\geq\cdots\geq v_d$, are the eigenvalues of $F$, $N$ and $V$, respectively. In our case $F$ could be represented as the sum of two $d\times d$ hermitian matrices,
\begin{equation*}
F=\sum_{n=1}^{d-1}G_{nn}+\frac1{\sqrt{2}}\begin{pmatrix}
  0 & 1-\ii &1-\ii& \cdots &1-\ii \\
  1+\ii & 0 & \cdots & \cdots &1-\ii \\
  1+\ii & \vdots & \ddots &\ddots &\vdots \\
  \vdots & \vdots & \ddots &\ddots &\vdots\\
    1+\ii & 1+\ii & \cdots &\cdots&0
 \end{pmatrix},
\end{equation*}
where $N=\sum_{n=1}^{d-1}G_{nn}$ is a diagonal matrix. The maximal and minimal eigenvalues of $N$ are given by $n_1=\sum_{n=1}^{d-1}1/{\sqrt{n(n+1)}}$ and $n_d=-\sqrt{(d-1)/{d}}$.
In order to bound the eigenvalues of $F$, using inequality~(\ref{Wineq}), we need to evaluate the eigenvalues of,
\begin{equation*}
V=\frac1{\sqrt{2}}\begin{pmatrix}
  0 & 1-\ii &1-\ii& \cdots &1-\ii \\
  1+\ii & 0 & \cdots & \cdots &1-\ii \\
  1+\ii & \vdots & \ddots &\ddots &\vdots \\
  \vdots & \vdots & \ddots &\ddots &\vdots\\
    1+\ii & 1+\ii & \cdots &\cdots&0 
 \end{pmatrix},
\end{equation*}
whose eigenvalues equation is
\begin{equation*}
{\rm Det}(\widetilde{V})\equiv
{\rm Det}\begin{pmatrix}
  -\sqrt{2}v & 1-\ii &1-\ii& \cdots &1-\ii \\
  1+\ii & -\sqrt{2}v & \cdots & \cdots &1-\ii \\
  1+\ii & \vdots & \ddots &\ddots &\vdots \\
  \vdots & \vdots & \ddots &\ddots &\vdots\\
    1+\ii & 1+\ii & \cdots &\cdots&-\sqrt{2}v 
 \end{pmatrix}=0.
\end{equation*}
To solve this equation we use the identity
\begin{equation}\label{detID}
{\rm Det}\begin{pmatrix}
  A & B  \\
  C & D 
\end{pmatrix}={\rm Det}(D){\rm Det}(A-BD^{-1}C),
\end{equation}
which holds for invertible $D$. By identifying
$A =-\sqrt{2}v$, $B=(1-\ii)(1,\ldots,1)=C^\dagger$ are vectors with $d-1$ components, and $D=\widetilde{V}_{d-1}$, where $\widetilde{V}_{d-1}$ is a $(d-1)\times(d-1)$ matrix with exactly the same structure as $\widetilde{V}$, we obtain the recursive relation, 
\begin{equation}\label{recursive}
{\rm Det}(\widetilde{V}_{k+1})={\rm Det}(\widetilde{V}_{k})(-\sqrt{2}v-B_{k}\widetilde{V}_{k}^{-1}C_{k}),\; k=2,\ldots,d-1,
\end{equation}
where $B_{k}$ is the same vector as $B$ but with $k$ elements. To solve this recursive relation we simply need to give the solution for the $k=1$,
\begin{equation*}
{\rm Det}(\widetilde{V}_{2})=-2v^2+2=0.
\end{equation*}
We therefore have an analytical solution (in a form of a recursive relation)  for the eigenvalues of $V$, hence, we can use them to bound the eigenvalues of $F$ according to Eq.~(\ref{Wineq}),
\begin{equation*}
-\sqrt{\frac{d-1}{d}}+v_d \leq
\begin{matrix}
\text{eigenvalues}\\
\text{of } F
 \end{matrix} \leq \sum_{n=1}^{d-1}\frac{1}{\sqrt{n(n+1)}} + v_1.
 \end{equation*}

Next we bound the eigenvalues of $R_{nm}=F-d(d+1)G_{nm}$. For $n=m$ with $n=1,\ldots,d-1$, 
\begin{equation}\label{pnn}
R_{nn}=\sum_{m=1}^{d-1}G_{mm}-\tilde{d}G_{nn}+V,
\end{equation}
where $\tilde{d}=d(d+1)$. The largest eigenvalue of the matrices $\sum_{m=1}^{d-1}G_{m,m}-\tilde{d}G_{n,n}$ is $\sum_{m=1}^{d-1}1/{\sqrt{m(m+1)}}+\tilde{d}\sqrt{(d-1)/{d}}$ while the smallest eigenvalue is $-\sqrt{(d-1)/{d}}-\tilde{d}\sum_{m=1}^{d-1}1/{\sqrt{m(m+1)}}$. Using inequality~(\ref{Wineq}), the eigenvalues of $R_{nn}$ for all $n=1,\ldots,d-1$ are bounded by
\begin{align*}
&\begin{matrix}
\text{eigenvalues}\\
\text{of } R_{nn}
 \end{matrix}\geq -\sqrt{\frac{d-1}{d}}-\tilde{d}\sum_{m=1}^{d-1}\frac{1}{\sqrt{m(m+1)}}+v_d,\nn
&\begin{matrix}
\text{eigenvalues}\\
\text{of } R_{nn}
 \end{matrix}\leq \sum_{m=1}^{d-1}\frac{1}{\sqrt{m(m+1)}}+\tilde{d}\sqrt{\frac{d-1}{d}}+v_1.
\end{align*}

To bound the eigenvalue of $R_{nm}$ with $n\neq m$ let us first look at $R_{12}$,
\begin{equation*}
R_{12}=\sum_{n=1}^{d-1}G_{nn}+\frac1{\sqrt{2}}\begin{pmatrix}
  0 & 1-\tilde{d}-\ii &1-\ii& \cdots &1-\ii \\
  1-\tilde{d}+\ii & 0 & \cdots & \cdots &1-\ii \\
  1+\ii & \vdots & \ddots &\ddots &\vdots \\
  \vdots & \vdots & \ddots &\ddots &\vdots\\
    1+\ii & 1+\ii & \cdots &\cdots&0
 \end{pmatrix}.
\end{equation*}
The eigenvalues equation of the non-diagonal matrix, 
\begin{equation*}
W=\frac1{\sqrt{2}}\begin{pmatrix}
  0 & 1-\tilde{d}-\ii &1-\ii& \cdots &1-\ii \\
  1-\tilde{d}+\ii & 0 & \cdots & \cdots &1-\ii \\
  1+\ii & \vdots & \ddots &\ddots &\vdots \\
  \vdots & \vdots & \ddots &\ddots &\vdots\\
    1+\ii & 1+\ii & \cdots &\cdots&0
 \end{pmatrix}
\end{equation*}
reads 
\begin{align}\label{detW}
&{\rm Det}(\widetilde{W})\nn&\equiv{\rm Det}\begin{pmatrix}
  -\sqrt{2}w & 1-\tilde{d}-\ii &1-\ii& \cdots &1-\ii \\
  1-\tilde{d}+\ii & -\sqrt{2}w & \cdots & \cdots &1-\ii \\
  1+\ii & \vdots & \ddots &\ddots &\vdots \\
  \vdots & \vdots & \ddots &\ddots &\vdots\\
    1+\ii & 1+\ii & \cdots &\cdots&-\sqrt{2}w
 \end{pmatrix}=0.
\end{align}
Except of the first row and column (i.e., the $(d-1)\times(d-1)$ inner matrix) of the above equation equals $\widetilde{V}_{d-1}$. Therefore to solve Eq.~(\ref{detW}), one can use the same recursive relation of Eq.~(\ref{recursive}) with the same definitions of $A,B,C$ and $D$ as before  up to the $d-1$th `level', and at the $d$th level make use of
\begin{equation*}
{\rm Det}(\widetilde{W})={\rm Det}(\widetilde{V}_{d-1})(-\sqrt{2}w-B_{d-1}\widetilde{V}_{d-1}^{-1}C_{d-1}),
\end{equation*}
with the $d-1$ components' vectors $B_{d-1}=(1-\tilde{d}-\ii,1-\ii,\cdots,1-\ii)=C_{d-1}^\dagger$. Thus, we find the maximum and minimum eigenvalues of $W$, $w_1$ and $w_d$, which enable us to bound the eigenvalues of $R_{12}$,
\begin{equation}\label{boundQ12}
-\sqrt{\frac{d-1}{d}}+w_d \leq
\begin{matrix}
\text{eigenvalues}\\
 \text{of }R_{12} 
 \end{matrix} \leq \sum_{n=1}^{d-1}\frac{1}{\sqrt{n(n+1)}} + w_1.
 \end{equation}
By writing $R_{nm}$ with $n\neq m$ as
\begin{equation*}
  R_{nm}=\sum_{k=1}^{d-1}G_{kk}+V-\tilde{d}G_{nm},
\end{equation*}
one can show that the eigenvalues of $W_{nm}=V-\tilde{d}G_{nm}$ are equal to the eigenvalues of $W$, and therefore the eigenvalues of  $R_{nm}$ are all bounded by the bounds appearing in Eq.~(\ref{boundQ12}). 

Upon defining 
\begin{align*}
\lambda_{\rm max}&={\max}\{(d+1)\Bigl( \sum_{n=1}^{d-1}\frac{1}{\sqrt{n(n+1)}} + v_1 \Bigr),\nn
&\;\;\;\;\;\; \sum_{m=1}^{d-1}\frac{1}{\sqrt{m(m+1)}}+\tilde{d}\sqrt{\frac{d-1}{d}}+v_1,\nn
&\;\;\;\;\;\; \sum_{n=1}^{d-1}\frac{1}{\sqrt{n(n+1)}} + w_1 \}\nn
\lambda_{\rm min}&={\min}\{(d+1)\Bigl( -\sqrt{\frac{d-1}{d}}+v_d \Bigr), -\sqrt{\frac{d-1}{d}}+w_d,\nn
&\;\;\;\;\;\; -\sqrt{\frac{d-1}{d}}-\tilde{d}\sum_{m=1}^{d-1}\frac{1}{\sqrt{m(m+1)}}+v_d,\},
\end{align*}
$t_0=-\frac1{d^2}\frac1{\lambda_{\rm max}}$, $t_1=-\frac1{d^2}\frac1{\lambda_{\rm min}}$, we obtain $t_m=\max\{|t_0|,t_1\}$. The optimal SIC POVM associated with the Gell-Mann basis is then given by
\begin{align*}
& P_{nm}\equiv\frac{1}{d^2}I+t_m\left(\sum_{k=1}^{d-1}G_{kk}- d(d-1) G_{nm}\right)\nonumber\\
& P_{dd}\equiv\frac{1}{d^2}I-t_m(d-1)\sum_{k=1}^{d-1}G_{kk} \;.\label{SICgm}
\end{align*}
We have numerically calculated $t_m$ as a function of the dimension and plotted in Fig.~(\ref{fig:tm}) the ratio of $t_m$ to its value had a rank-1 SIC POVM exists in this dimension $1/(d(d+1))^{3/2}$. We found that the as the dimension grows $t_m(d(d+1))^{3/2}$ quickly drops to zero, that is the Gell-Mann basis is far from an optimal basis. 
\begin{figure}[t]
\centering
\includegraphics[width=.8\linewidth]{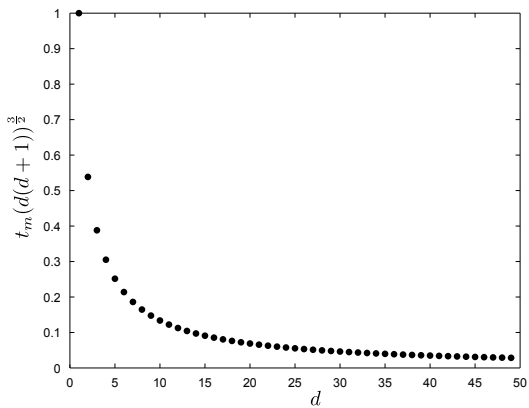}
\caption{A plot of $t_m(d(d+1))^{3/2}$ as a function of the dimension $d$. The $t_m$ is calculated for the Gell-Mann basis, and the quantity $0\leq t_m(d(d+1))^{3/2}\leq 1$ indicates how close the SIC POVM to a rank-1 POVM (upper bound). As the dimension grows $t_m(d(d+1))^{3/2}$ quickly drops to zero, that is the Gell-Mann basis results in a SIC POVM whose elements are close to the identity operator.}
\label{fig:tm}
\end{figure}%

\section{Rank 1 SIC POVMs}\label{sec:rank1}
From the expression for $a$, Eq.~(\ref{at}), we see that if $a=1/d^2$ then
\begin{equation}\label{t}
t=\frac{1}{[d(d+1)]^{3/2}}\;,
\end{equation}
where we assumed without loss of generality that $t$ is positive since we can always replace the orthonormal basis $\{F_\alpha\}$ with the orthonormal basis $\{-F_{\alpha}\}$. Substituting this value of $t$ in Eq.~(\ref{f}) gives:
\begin{equation}\label{rank1f}
F_\alpha=\frac{1}{\sqrt{d(d+1)}}\left(I+\Pi_{d^2}-(d+1)\Pi_\alpha\right)\;,
\end{equation}
for rank 1 SIC POVMs where $\Pi_{\alpha}\equiv dP_{\alpha}$ and $\Pi_{d^2}\equiv dP_{d^2}$ are rank 1 projections.

If $\{P_\alpha\}$ is a rank 1 SIC POVM then the eigenvalues of $F_\alpha$ in~(\ref{rank1f}) are not depending on $\alpha$
since $\tr(\Pi_{\alpha}P_{d^2})=1/(d+1)$ for all $\alpha=1,2,...,d^2-1$. A straightforward calculation shows that the eigenvalues
of the rank 2 matrix $\Pi_{d^2}-(d+1)\Pi_\alpha$ are given by $\left(-d\pm\sqrt{d^2+4d}\right)/2$. Thus, for rank 1 SIC POVMs the eigenvalues of all $F_\alpha$ are given by:
\begin{align}
& \gamma_1=\frac{1}{2\sqrt{d(d+1)}}\left(2-d+\sqrt{d^2+4d}\right)\nonumber\\
& \gamma_2=\frac{1}{2\sqrt{d(d+1)}}\left(2-d-\sqrt{d^2+4d}\right)\nonumber\\
& \gamma_k=\frac{1}{\sqrt{d(d+1)}}\;\;\;\forall\;k=3,...,d\;.\label{lambda}
\end{align}
This observation indicates that while the class of general SIC POVMs in $\rH_d$ is relative big (any orthonormal basis $\{F_\alpha\}$ yields a general SIC POVM), the class of rank 1 SIC POVMs in $\rH_d$ (if exists) is extremely small. The fact that all
$\{F_\alpha\}$ have the same eigenvalues implies that there exist $d^2-1$ unitary matrices $\{U_\alpha\}$ such that
$F_\alpha=U_\alpha D U_{\alpha}^{\dag}$, where $D={\rm diag}\{\gamma_1,...,\gamma_d\}$ is a diagonal matrix with the $\gamma_k$s as in Eq.~(\ref{lambda}). The orthogonality relation reads as 
$\tr(U_\alpha D U_{\alpha}^{\dag}U_\beta D U_{\beta}^{\dag})=\delta_{\alpha,\beta}$. Thus, if the set $\{U_\alpha\}$ forms a group, 
the orthogonality relations can be written simply as $\tr(U_\alpha D U_{\alpha}^{\dag}D)=0$ for $U_\alpha\neq I$.
This is somewhat reminiscent to the more standard construction of rank 1 SIC POVMs with the Weyl-Heisenberg group~\cite{Rene04,Sco09}.

\section{Example in $d=2$}\label{sec:d2}

The simplest example to construct is in dimension $d=2$. In this case, we take the three normalized Pauli matrices 
$$
F_1\equiv\frac{1}{\sqrt{2}}\left(\begin{array}{cc}0&1\\1&0\end{array}\right),
F_2\equiv\frac{1}{\sqrt{2}}\left(\begin{array}{cc}0&i\\-i&0\end{array}\right),
F_3\equiv\frac{1}{\sqrt{2}}\left(\begin{array}{cc}1&0\\0&-1\end{array}\right)
$$
to be our fixed orthonormal basis for the vector space of $2\times 2$ traceless Hermitian matrices $\rT_2$.
With this choice we get
$$
F=\frac{1}{\sqrt{2}}\left(\begin{array}{cc}1&1+i\\1-i&-1\end{array}\right)\;.
$$
The eigenvalues of the three matrices $F-d(d+1)F_\alpha=F-6F_\alpha$ ($\alpha=1,2,3$) are all the same
and are given by $\lambda_{\alpha}=3\sqrt{3/2}$ and $\mu_\alpha=-3\sqrt{3/2}$. Also the eigenvalues of
$3F$ are the same and given by $\lambda_4=3\sqrt{3/2}$ and $\mu_4=-3\sqrt{3/2}$. Thus, $t_0=-\frac{1}{12}\sqrt{\frac{2}{3}}$,
 $t_1=\frac{1}{12}\sqrt{\frac{2}{3}}$, and for any non-zero $t\in[t_0,t_1]$ the four matrices
 \begin{align*}
 & P_{\alpha}\equiv\frac{1}{4}I+t\left(F- 6 F_\alpha\right)\;\;\text{for}\;\alpha=1,2,3\\
& P_{4}\equiv\frac{1}{4}I+3tF
 \end{align*}
 form a general SIC POVM. Note that $t_1$ equals to the maximum possible value given in Eq.~(\ref{t}).
 Thus, for $t=t_1$ we get the following rank 1 SIC POVM:
 \begin{align*}
 & P_{1}=\frac{1}{12\sqrt{3}}\left(\begin{array}{cc}3\sqrt{3}+1&-5+i\\-5-i&3\sqrt{3}-1\end{array}\right)\\
&P_{2}=\frac{1}{12\sqrt{3}}\left(\begin{array}{cc}3\sqrt{3}+1&1-5i\\1+5i&3\sqrt{3}-1\end{array}\right)\\
&P_{3}=\frac{1}{12\sqrt{3}}\left(\begin{array}{cc}3\sqrt{3}-5&1+i\\1-i&3\sqrt{3}+5\end{array}\right)\\
&P_{4}=\frac{1}{4\sqrt{3}}\left(\begin{array}{cc}\sqrt{3}+1&1+i\\1-i&\sqrt{3}-1\end{array}\right)\\
 \end{align*}
Note that the above rank 1 SIC POVM is equivalent (up to a global rotation) to the original one introduced first in~\cite{Caves,Zau99}. Our example though shows that the rank 1 SIC POVM can be obtained from the three Pauli matrices. Thus, the three normalized Pauli matrices form an optimal basis for $\rT_{d=2}$. The generalized Gell-Mann matrices, which reduce to the normalized Pauli matrices for $d=2$, are not the optimal basis for $d>2$, thus do not correspond rank 1 SIC POVM for higher dimensions, as indicated in Fig.~\ref{fig:tm}.

\section{Discussion}\label{sec:conc}
We have constructed the complete set of general SIC POVMs given in Eq.~(\ref{main}). Our construction shows that the family of general SIC POVMs is as big as the set of all orthonormal bases in $\rT_d$. Fixing an orthonormal basis $\{F_\alpha\}$, every other orthonormal basis $\{F'_\alpha\}$ of $\rT_d$ can be obtained from $\{F_\alpha\}$ via $F'_\alpha=\sum_{\beta}A_{\alpha\beta}F_\beta$, where $A$ is a real orthogonal matrix. Thus, every element in the group $O(d^2-1)$ (i.e. the set of $(d^2-1)\times (d^2-1)$ real orthogonal matrices) defines a one parameter set of general SIC POVMs (see Eq.(\ref{main})). Typically, distinct orthogonal matrices corresponds to distinct SIC POVMs.

The SIC POVMs defined in Eq.(\ref{main}) can be made arbitrarily close to 
$\frac{1}{d^2}I$ by taking $t$ to be close enough to zero. These weak SIC POVMs do not disturb much the state of the system and therefore can be used for tomography that is followed by other quantum information processing tasks.  Note that from Eq.~(\ref{main}) \emph{any} orthonormal basis of $\rT_d$ can be used to construct a weak SIC POVM. This implies that the set of weak SIC POVMs is extremely big whereas the set of rank 1 SIC POVMs (if exists) is extremely small.

The purity parameter $a$ of a general SIC POVM determines how close it is to a rank 1 SIC POVM. If $a$ is very close to $1/d^3$ than the general SIC POVMs is \emph{full} rank, whereas if it is very close to $1/d^2$ it is close to being rank 1. From our construction it is obvious that if there exists a SIC POVM with some $a=a_0>1/d^3$, then there exist SIC POVMs with \emph{any} $a$ in the range $(1/d^3,a_0]$. Thus, it is natural to look for the highest possible value that $a$ can take which we denoted in Eq.~(\ref{amax}) by $a_{\max}$. From both the analytical and numerical evidence we know that $a_{\max}=1/d^2$ for low dimensions since in small dimensions there exists rank 1 SIC POVMs~\cite{Sco09}. However, since in higher dimensions it seems to be a hard task to construct rank 1 SIC POVMs, one can try to find a lower bound for $a_{\max}$. In Sec.~\ref{sec:gm} considered such bound by constructing SIC POVMs through the generalized Gell-Mann operator basis. As one may expect, we found that the lower bound quickly approaches its lower value $1/d^3$ as the dimension increases. However, one can numerically search for bases which give interesting bounds by making orthogonal transformations on the  generalized Gell-Mann basis, and to look for those transformations for which $a$ increases. We leave this analysis for future work. 

\emph{Acknowledgments:---}
We would like to thank the anonymous referees that suggested to improve the derivation of Theorem~\ref{thm:main} and to consider the particular case of SIC POVMs constructed from the generalized Gell-Mann basis.  G.G. thanks Rob Spekkens and Nolan Wallach for many fruitful discussions that initiated this project, to Markus Grassl for helpful correspondence and comments on the first draft of this paper, and to Christopher Fuchs for pointing out reference~\cite{App07}. G.G. research was partially supported by NSERC and by the Air Force Office of Scientific Research as part of the Transformational Computing in Aerospace Science and Engineering Initiative under grant FA9550-12-1-0046. A.K. research is supported by NSF Grant PHY-1212445. 

\appendix
\section{Proof of Lemma~\ref{le:symOp}}\label{app:proof_le:symOp}
\begin{proof}
To prove that the operators $P_\alpha$, $\alpha=1,2,...,d^2$ of Eq.~(\ref{symOp}) are symmetric according to Definition~\ref{def:sic}, we first note that $\tr(P_\alpha)=1/d$ for all $\alpha=1,2,...,d^2$.
For $\alpha,\beta\in\{1,2,...,d^2-1\}$ we have
\begin{align*}
\tr(P_\alpha P_\beta)&=\frac{1}{d^3}+t^2\tr(R_\alpha R_\beta)\\
&=\frac{1}{d^3}+t^2r\delta_{\alpha,\beta}-\frac{t^2r}{d^2-1}(1-\delta_{\alpha,\beta}).
\end{align*}
We also have for $\alpha\in\{1,2,...,d^2-1\}$
\begin{align*}
\tr(P_\alpha P_{d^2})&=\frac{1}{d^3}-t^2\left(r-\frac{r}{d^2-1}(d^2-2)\right)\\
&=\frac{1}{d^3}-\frac{t^2r}{d^2-1}\;,
\end{align*}
and
\begin{align*}
\tr(P_{d^2} P_{d^2})&=\frac{1}{d^3}+t^2\left((d^2-1)r-(d^2-1)(d^2-2)\frac{r}{d^2-1}\right)\\
&=\frac{1}{d^3}+t^2r\;.
\end{align*}
Thus, all that is left to show is that $\{P_\alpha\}_{\alpha=1,2,...,d^2}$ is linearly independent.
Indeed, suppose
$\sum_{\alpha=1}^{d^2}s_\alpha P_\alpha=0$. Note first that by taking the trace on both sides we get $\sum_{\alpha=1}^{d^2}s_\alpha=0$. Thus,
$$
0=\sum_{\alpha=1}^{d^2}s_\alpha P_\alpha=t\sum_{\alpha=1}^{d^2-1}\left(s_\alpha-s_{d^2}\right)R_\alpha\;.
$$
Recall that $t\neq 0$ and $R_\alpha$ are $d^2-1$ linearly independent matrices in $\rH_d$. Thus, $s_\alpha=s_{d^2}$,
and since $\sum_{\alpha=1}^{d^2}s_\alpha=0$ we get $s_\alpha=0$ for all $\alpha=1,2,...,d^2$.
\end{proof}

\end{document}